\documentclass[a4paper]{article}

\usepackage{tcolorbox}
\usepackage{xspace}
\usepackage{todonotes}
\usepackage{mathtools}
\usepackage{amsmath, amssymb, amsthm}
\usepackage{thm-restate}
\usepackage{hyperref}
\usepackage{cleveref}
\usepackage{a4wide}
\usepackage[T1]{fontenc}
\usepackage[utf8]{inputenc}
\usepackage{textpos}

\title{Detecting Points in Integer Cones of Polytopes is Double-Exponentially Hard\thanks{This work is a part of project BOBR (ŁK, KM, MP, MS) that has received funding from the European Research Council (ERC) under the European Union’s Horizon 2020 research and innovation programme (grant agreement No. 948057).
A. Lassota was supported by the Swiss National Science Foundation within the project \emph{Complexity of Integer Programming}~(207365).}}

\author{blah}

\author{
	\L{}ukasz Kowalik\thanks{Institute of Informatics, University of Warsaw, Poland (\texttt{kowalik@mimuw.edu.pl})}
	\and Alexandra Lassota\thanks{Institute of Mathematics, EPFL, Lausanne, Switzerland ({\texttt{alexandra.lassota@epfl.ch}})}
	\and Konrad Majewski\thanks{Institute of Informatics, University of Warsaw, Poland ({\texttt{k.majewski@mimuw.edu.pl}})}
	\and Michał Pilipczuk\thanks{Institute of Informatics, University of Warsaw, Poland ({\texttt{michal.pilipczuk@mimuw.edu.pl}})}
	\and Marek Sokołowski\thanks{Institute of Informatics, University of Warsaw, Poland ({\texttt{marek.sokolowski@mimuw.edu.pl}})}
}

\newtheorem{theorem}{Theorem}
\newtheorem{claim}{Claim}

\newcommand{\Inline}[1]{#1\xspace}
\newcommand{\ProblemFormat}[1]{\textsc{#1}}
\newcommand{\ProblemName}[1]{\Inline{\ProblemFormat{#1}}}

\newcommand{\Oh}{\mathcal{O}}
 
 \newcommand{\cqed}{\ensuremath{\lhd}}
 \makeatletter
 \newenvironment{claimproof}{\par
         \pushQED{\cqed}%
         \normalfont \topsep6\p@\@plus6\p@\relax
         \trivlist
         \item\relax
         {\itshape
                 Proof of the claim\@addpunct{.}}\hspace\labelsep\ignorespaces
 }{%
         \hfill\popQED\endtrivlist\@endpefalse
 }
 \makeatother

\renewcommand{\leq}{\leqslant}
\renewcommand{\geq}{\geqslant}

\renewcommand{\le}{\leqslant}
\renewcommand{\ge}{\geqslant}
\renewcommand{\setminus}{-}

\newcommand{\ceil}[1]{\lceil #1 \rceil}

\newcommand{\ThreeSAT}{\ProblemName{3-SAT}}
\newcommand{\SubsetSum}{\ProblemName{Subset Sum}}
\newcommand{\SubsetSumRep}{\ProblemName{Subset Sum with Multiplicities}}

\newcommand{\PointInCone}{\ProblemName{Point in Cone}}
\newcommand{\PolytopeIntersectCone}{\ProblemName{Cone and Polytope Intersection}}
\newcommand{\BinPacking}{\ProblemName{Bin Packing}}

\newcommand{\yesinstance}{\textsc{Yes}-instance\xspace}

\newcommand{\RR}{\mathbb{R}}
\newcommand{\ZZ}{\mathbb{Z}}
\newcommand{\onevector}{\mathbf{1}}
\newcommand{\polyfont}[1]{\mathcal{#1}}
\newcommand{\poly}[1]{\polyfont{#1}}

\newcommand{\IntCone}{\mathsf{IntCone}}
\newcommand{\encoding}{\mathsf{enc}}

\newcommand{\Ic}{\mathcal{I}}

\newif\ifcomment
\commentfalse
\commenttrue

\begin{document}


%

\maketitle

\begin{abstract}
Let $d$ be a positive integer.
For a finite set $X \subseteq \RR^d$, we define its \emph{integer cone} as the set $\mathsf{IntCone}(X) \coloneqq \{ \sum_{x \in X} \lambda_x \cdot x \mid \lambda_x \in \mathbb{Z}_{\geq 0} \} \subseteq \RR^d$.
Goemans and Rothvoss showed that, given two polytopes $\mathcal{P}, \mathcal{Q} \subseteq \RR^d$ with $\mathcal{P}$ being bounded, one can decide whether $\mathsf{IntCone}(\mathcal{P} \cap \mathbb{Z}^d)$ intersects $\mathcal{Q}$ in time $\mathsf{enc}(\mathcal{P})^{2^{\mathcal{O}(d)}} \cdot \mathsf{enc}(\mathcal{Q})^{\mathcal{O}(1)}$~[J.~ACM 2020], where $\mathsf{enc}(\cdot)$ denotes the number of bits required to encode a polytope through a system of linear inequalities. This result is the cornerstone of their $\mathsf{XP}$ algorithm for \BinPacking parameterized by the number of different item sizes.

We complement their result by providing a conditional lower bound.
In particular, we prove that, unless the ETH fails, there is no algorithm which, given a bounded polytope $\mathcal{P} \subseteq \RR^d$ and a point $q \in \ZZ^d$, decides whether $q \in \mathsf{IntCone}(\mathcal{P} \cap \ZZ^d)$ in time $\mathsf{enc}(\mathcal{P}, q)^{2^{o(d)}}$. Note that this does {\em{not}} rule out the existence of a fixed-parameter tractable algorithm for the problem, but shows that dependence of the running time on the parameter $d$ must be at least doubly-exponential.

\end{abstract}

 \begin{textblock}{20}(-1.9, 4.2)
	\includegraphics[width=40px]{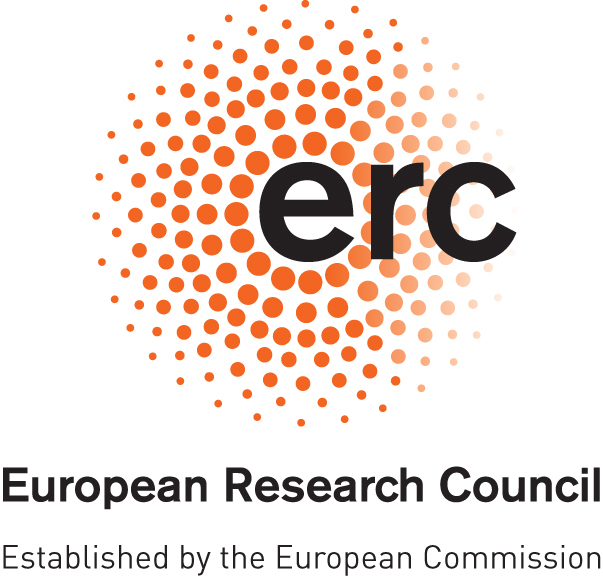}%
\end{textblock}
\begin{textblock}{20}(-2.15, 4.5)
	\includegraphics[width=60px]{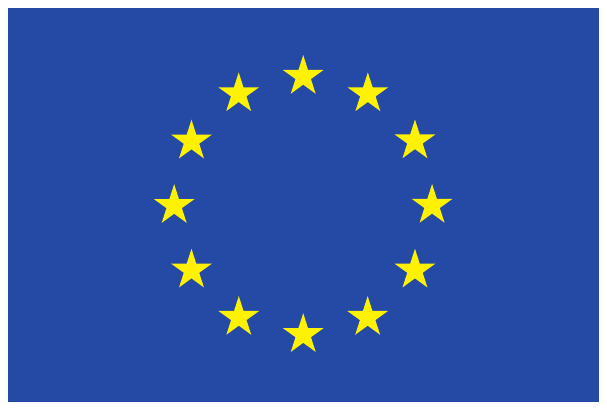}%
\end{textblock}

\section{Introduction}
\label{sec:intro}

Consider the following high-multiplicity variant of the \BinPacking problem: given a vector~$s=(s_1,\ldots,s_d)\in [0,1]^d$ of  item sizes and a vector of multiplicities $a=(a_1,\ldots,a_d)\in \ZZ_{\geq 0}^d$, find the smallest integer $B$ so that the collection of items containing $a_i$ items of size $s_i$, for each $i\in \{1,\ldots,d\}$, can be entirely packed into $B$ unit-size bins. In their celebrated work~\cite{GoemansR20}, Goemans and Rothvoss gave an algorithm for this problem with time complexity~$\encoding(s,a)^{2^{\Oh(d)}}$, where $\encoding(s,a)$ denotes the total bitsize of the encoding of $s$ and $a$ in binary. In the terminology of parameterized complexity, this puts high-multiplicity \BinPacking parameterized by the number of different item sizes in the complexity class $\mathsf{XP}$.

In fact, Goemans and Rothvoss studied the more general \PolytopeIntersectCone problem, defined as follows: given two polytopes $\poly{P},\poly{Q}\subseteq \RR^d$, where $\poly{P}$ is bounded, is there a point in $\poly{Q}$ that can be expressed as a nonnegative integer combination of integer points in $\poly{P}$? Goemans and Rothvoss gave an algorithm for this problem with running time~$\encoding(\poly{P})^{2^{\Oh(d)}}\cdot \encoding(\poly{Q})^{\Oh(1)}$, where $\encoding(\poly{R})$ denotes the total bitsize of the encoding of a polytope~$\poly{R}$ through a system of linear inequalities. They showed that high-multiplicity \BinPacking admits a simple reduction to \PolytopeIntersectCone, where in essence, integer points in $\poly{P}$ correspond to possible configurations of items that fit into a single bin and $\poly{Q}$ is the point corresponding to all items (more precisely, $\poly{P}=\{({x\atop 1})\in \RR^{d+1}_{\ge 0}\mid s^Tx\le 1\}$ and $\poly{Q}=\{({a\atop B})\}$). 
In fact, \PolytopeIntersectCone is a much more versatile problem: in \cite[Section~6]{GoemansR20}, Goemans and Rothvoss present a number of applications of their result to other problems in the area of scheduling.

Whether the result of Goemans and Rothvoss for high-multiplicity \BinPacking can be improved to fixed-parameter tractability is considered a major problem in the area. It was asked already by Goemans and Rothvoss in~\cite{GoemansR20}, addressed again by Jansen and Klein in~\cite{JansenK20}, and also discussed in the survey of Mnich and van Bevern~\cite{MnichB18}.

In this work, we take a step into solving the complexity of this problem. We prove the following result that shows that the doubly-exponential dependence on $d$ in the running times of algorithms for \PolytopeIntersectCone is necessary, assuming the Exponential-Time Hypothesis (ETH). The lower bound holds even for the simpler \PointInCone problem, where the polytope $\poly{Q}$ consists of a single integer point $q\in \ZZ^d$.

\begin{theorem}
\label{thm:main}
Unless the ETH fails, there is no algorithm solving \PointInCone in time~$\encoding(\poly{P}, q)^{2^{o(d)}}$, where $\encoding(\poly{P}, q)$ is the total number of bits required to encode both $\poly{P}$ and~$q$.
\end{theorem}

Notice that \cref{thm:main} does {\em{not}} rule out the possibility that there exists a fixed-parameter algorithm with running time $f(d)\cdot \encoding(\poly{P},q)^{\Oh(1)}$ for some function $f$. However, it shows that for this to hold, function $f$ would need to be at least doubly exponential, assuming the ETH.

Let us briefly elaborate on our proof of \cref{thm:main} and its relation to previous work. The cornerstone of the result of Goemans and Rothvoss is a statement called Structure Theorem, which essentially says the following: if an instance of \PolytopeIntersectCone has a solution, then it has a solution whose {\em{support}} --- the set of integer points in $\poly{P}$ participating in the nonnegative integer combination yielding a point in~$\poly{Q}$ --- has size at most $2^{2d+1}$. Moreover, except for a few outliers, this support is contained within a carefully crafted set $X$ consisting of roughly $\encoding(\poly{P})^{\Oh(d)}$ integer points within $\poly{P}$. In subsequent work~\cite{JansenK20}, Jansen and Klein showed a more refined variant of the Structure Theorem where $X$ is just the set of vertices of the convex hull of the integer points lying in $\poly{P}$; but the exponential-in-$d$ bound on the size of the support persists. The appearance of this bound in both works~\cite{GoemansR20,JansenK20} originates in the following elegant observation of Eisenbrand and Shmonin~\cite{EisenbrandS06}: whenever some point $v$ can be represented as a nonnegative integer combination of integer points in $\poly{P}$, one can always choose such a representation of $v$ with support of size bounded by $2^d$ (see~\cite[Lemma~3.4]{GoemansR20} for a streamlined proof). In~\cite[Section~8]{GoemansR20}, Goemans and Rothvoss gave an example showing that the $2^d$ bound is tight up to a multiplicative factor of~$2$, thereby arguing that within their framework, one cannot hope for any substantially better bound on the support size. The main conceptual contribution of this work can be expressed as follows: the construction showing the tightness of the observation of Eisenbrand and Shmonin not only exposes a bottleneck within the support-based approach of~\cite{GoemansR20,JansenK20}, but in fact can be used as a gadget in a hardness reduction proving that the doubly-exponential dependence on $d$ in the running time is necessary for the whole problem, assuming the ETH.

Finally, we remark that tight doubly-exponential lower bounds under the ETH appear scarcely in the literature, as in reductions proving such lower bounds, the parameter of the output instance has to depend logarithmically on the size of the input instance of \ThreeSAT. A few examples of such lower bounds can be found here:~\cite{CyganPP16,FominGLSZ19,JansenKL23,KnopPW20,KunnemannMSSBW23,MarxM16}; our work adds \PointInCone to this rather exclusive list.

\section{Preliminaries}
\label{sec:prelims}

For a positive integer $n$, we denote $[n] \coloneqq \{1, 2, \ldots, n\}$ and $[n]_0 \coloneqq \{0, 1, \ldots, n-1\}$.

\subparagraph*{Euclidean spaces.}
Fix a positive integer $d$.
We call the elements of $\RR^d$ \emph{vectors} (or \emph{points}).
Given a vector $x \in \RR^d$, we denote its $i$-th coordinate (for $i \in [d]$) by $x(i)$.
By $\onevector_d$, we denote the $d$-dimensional vector of all ones, that is, $\onevector_d = (1, \ldots, 1)\in \ZZ^d$.
When the dimension $d$ is clear from the context, we omit it from the subscript and simply write $\onevector$ instead.

We allow vectors to be added to each other, and to be multiplied by a scalar $\lambda \in \RR$.
Both operations come from treating the space $\RR^d$ as a linear space over $\RR$.
Given a finite set~$X \subseteq \RR^d$, we define its \emph{integer cone} as the set
\[
\IntCone(X) \coloneqq \left\{ \sum_{x \in X} \lambda_x \cdot x \mid \lambda_x \in \ZZ^d_{\geq 0} \text{ for  every } x \in X \right\}. 
\]

\subparagraph*{Polytopes.}
In this work a \emph{$d$-dimensional polytope} is a subset of points in $\RR^d$ satisfying a system of linear inequalities with integer coefficients, that is, a set of the form $\poly{P} \coloneqq \{ x \in \RR^d \mid Ax \leq b \}$, where $A \in \ZZ^{d \times m}$ and $b \in \ZZ^m$ for some positive integer $m$.
Then, the \emph{encoding size of $\poly{P}$}, denoted $\encoding(\poly{P})$, is the total number of bits required to encode the matrix~$A$ and the vector~$b$.
We say that the polytope $\poly{P}$ is \emph{bounded} if there exists a number $M \in \ZZ$ such that for all $x \in \poly{P}$ and $i\in [d]$, we have $|x(i)| \leq M$.

We can now define the main problem studied in this paper, namely \PointInCone.

\smallskip

\begin{tcolorbox}

\smallskip

\noindent \PointInCone

\smallskip

\noindent \textbf{Input:} A positive integer $d$, a bounded polytope $\poly{P} \subseteq \RR^d$ (given by a matrix $A \in \ZZ^{m \times d}$ and a vector $b \in \ZZ^m$ for some integer $m$), and a point $q \in \ZZ^d$.

\smallskip

\noindent \textbf{Question:} Is $q \in \IntCone(\poly{P} \cap \ZZ^d)$?

\end{tcolorbox}

As mentioned in~\cref{sec:intro}, in~\cite{GoemansR20} Goemans and Rothvoss gave an algorithm for \PointInCone that runs in time $\encoding(\poly{P})^{2^{\Oh(d)}} \cdot \encoding(q)^{\Oh(1)}$. In fact, they solved the more general \PolytopeIntersectCone, where instead of a single point $q$, we are given a polytope $\poly{Q}$, and the question is whether $\IntCone(\poly{P} \cap \ZZ^d)\cap \poly{Q}$ is nonempty. In this case, the running time is $\encoding(\poly{P})^{2^{\Oh(d)}} \cdot \encoding(\poly{Q})^{\Oh(1)}$.

%

\subparagraph*{ETH.}
The Exponential-Time Hypothesis (ETH), proposed by Impagliazzo et al.~\cite{ImpagliazzoPZ01}, plays a fundamental role in providing conditional lower bounds for parameterized problems. It postulates that there exists a constant $c>0$ such that the \ThreeSAT problem cannot be solved in time $\Oh(2^{cn})$, where $n$ is the number of variables of the input formula. As proved in~\cite{ImpagliazzoPZ01}, this entails that there is no algorithm for \ThreeSAT with running time $2^{o(n+m)}$, where $m$ denotes the number of clauses of the input formula; see also~\cite[Theorem~14.4]{platypus}. We refer the reader to \cite[Chapter~14]{platypus} for a thorough introduction to ETH-based lower bounds within parameterized complexity.

\subparagraph*{Subset Sum.}
The classic \SubsetSum problem asks, for a given set $S$ of positive integers and a target integer $t$, whether there is a subset $S' \subseteq S$ such that $\sum_{x \in S'} = t$.
The standard $\mathsf{NP}$-hardness reduction from \ThreeSAT to \SubsetSum takes an instance of \ThreeSAT with $n$ variables and $m$ clauses and outputs an equivalent instance of \SubsetSum where $|S|=\Oh(n+m)$ and $t\leq 2^{\Oh(n+m)}$. By combining this with the $2^{o(n+m)}$-hardness for \ThreeSAT following from ETH, we obtain the following.

\begin{theorem}
\label{thm:subsetsum}
Unless the ETH fails, there is no algorithm solving \SubsetSum in time~$2^{o(n)}$, even under the assumption that $t\leq 2^{\Oh(n)}$. Here, $n$ denotes the cardinality of the set $S$ given on input.
\end{theorem}

In this work, we rely on a variant of the \SubsetSum problem called \SubsetSumRep.
The difference between those two problems is that in the latter one, we allow the elements from the input set to be taken with any nonnegative multiplicities.

\begin{tcolorbox}

\noindent \SubsetSumRep

\smallskip

\noindent \textbf{Input:} A set of positive integers $\{a_1, a_2, \ldots, a_n\}$ and a positive integer $t$.

\smallskip

\noindent \textbf{Question:} Does there exist a sequence of $n$ nonnegative integers $(\lambda_1, \lambda_2, \ldots, \lambda_n)$ such that $
\sum_{i=1}^n \lambda_i \cdot a_i = t?$

\end{tcolorbox}

The same lower bound as in Theorem~\ref{thm:subsetsum} holds for \SubsetSumRep.
This can be shown via a~simple reduction from \SubsetSum.
As this is standard, we present the proof of the following Theorem~\ref{thm:subsetsumrep} in Appendix~\ref{app:subsetsumrep}.

\begin{restatable}{theorem}{subsetsumrepthm}
\label{thm:subsetsumrep}
Unless the ETH fails, there is no algorithm solving \SubsetSumRep in time~$2^{o(n)}$, even under the assumption that $t\leq 2^{\Oh(n)}$. Here, $n$ denotes the cardinality of the set given on input.
\end{restatable}

\section{Reduction}
\label{sec:reduction}

The entirety of this section is devoted to the proof of our double-exponential hardness result: \cref{thm:main}.

The proof is by reduction from \SubsetSumRep.
Let $\Ic = (\{a_1, \ldots, a_n\}, t)$ be the given instance of \SubsetSumRep. That is, we ask whether there are nonnegative integers $\lambda_1, \ldots, \lambda_n$ such that $\sum_{i=1}^n \lambda_i \cdot a_i = t$, where $a_1, \ldots, a_n, t$ are given positive integers. We may assume that $a_i\leq t$ for all $i\in [n]$ and, following on the hardness postulated by \cref{thm:subsetsumrep}, that $t\leq 2^{\Oh(n)}$.

Let $d \coloneqq \lceil \log_2(n + 1) \rceil + 1$, hence $d$ satisfies the inequality $2^d \geq 2n + 2$ and $d = \Oh(\log n)$.
Let $\chi_0, \chi_1, \ldots, \chi_{2^d-1} \in \ZZ^d$ be all $\{0,1\}$-vectors in $d$-dimensional space, listed in lexicographic order.
Equivalently, $\chi_i$ is the bit encoding of the number $i$, for $i \in [2^d]_0$.
Observe that we have $\chi_i + \chi_{2^d-1-i} = \onevector$ for every $i \in [2^d]_0$.

We define the set $P \coloneqq \{p_0, p_1, \ldots, p_{2^d-1}\} \subseteq \ZZ^{d+1}$ of $2^d$ points as follows.
\[
p_i(j) = \begin{cases}
\chi_i(j), & \text{ for } i \in [2^d]_0 \text{ and } j \in [d]; \\
a_i, & \text{ for } i \in [n] \text{ and } j = d+1;\\
0 & \text{ for } i \in[2^d]_0\setminus [n] \text{ and } j = d+1.
\end{cases}
\]
We remark that the construction of the point set $P$ is inspired by the example of Goemans and Rothvoss provided in~\cite[Section~8]{GoemansR20}. First, we argue that $P$ can be expressed as integer points in a polytope of small encoding size.


\begin{claim}
\label{cl:polytope}
There exists a bounded polytope $\poly{P}$ of encoding size $\Oh(n \log n \cdot \log t)$ such that $\poly{P} \cap \ZZ^{d+1} = P$.
\end{claim}

\begin{claimproof}
Let $\poly{P}$ be the polytope defined by the following inequalities.
\begin{gather}
  0 \leq x(j) \leq 1 \quad\text{for } j \in [d], \label{eq:bound-first-coordinates} \\
  0 \leq x(d+1) \leq t, \label{eq:bound-last-coordinate} \\
x(d+1)+\sum_{j\colon \chi_i(j) = 0} t \cdot x(j)  + \sum_{j\colon \chi_i(j) = 1} t \cdot (1 - x(j)) \geq p_i(d+1) \quad\text{for } i \in [2^d]_0, \label{eq:ub}\\
t-x(d+1)+\sum_{j\colon \chi_i(j) = 0} t \cdot x(j) + \sum_{j\colon \chi_i(j) = 1} t \cdot (1 - x(j))  \geq t-p_i(d+1)  \text{ for } i \in [2^d]_0 \label{eq:lb}
\end{gather}

By~\eqref{eq:bound-first-coordinates} and~\eqref{eq:bound-last-coordinate},  $\poly{P}$ is bounded. Also, encoding the system of all linear inequalities defining $\poly{P}$ takes
\[
\Oh(2^d \cdot d \cdot \log t) = \Oh(n \log n \cdot \log t)
\]
bits, as desired. It remains to show that $\poly{P} \cap \ZZ^{d+1} = P$.
In what follows, when $i\in  [2^d]_0$, $(\ref{eq:ub}.i)$ denotes the single inequality of the form~\eqref{eq:ub} for this particular $i$, similarly for inequalities of the form~\eqref{eq:lb}.

First we show $\poly{P} \cap \ZZ^{d+1} \subseteq P$.
Pick $x\in \poly{P} \cap \ZZ^{d+1}$.
Since $x\in \ZZ^{d+1}$ and $x$ satisfies~\eqref{eq:bound-first-coordinates}, the first $d$ coordinates of $x$ form a binary encoding of a number ${i}^\star\in [2^d]_0$. Then, $x(j) = \chi_{i^\star}(j)$ for $j \in [d]$, hence by $(\ref{eq:ub}.i^\star)$, $x(d+1)  \ge  p_{i^\star}(d+1)$ and by 
$(\ref{eq:lb}.i^\star)$, $x(d+1)  \le  p_{i^\star}(d+1)$.
It follows that $x(d+1)  =  p_{i^\star}(d+1)$ and hence $x=p_{i^\star}\in P$, as required.

Finally we show $P \subseteq\poly{P} \cap \ZZ^{d+1}$.
Pick $x=p_{i^\star}\in P$ for some ${i}^\star\in [2^d]_0$. We need to show that~\eqref{eq:bound-first-coordinates}--\eqref{eq:lb} hold for $x$.
This is clear for~\eqref{eq:bound-first-coordinates} and~\eqref{eq:bound-last-coordinate}.
The inequality $(\ref{eq:ub}.i^\star)$ for $x$ is just $x(d+1)\ge p_{i^\star}(d+1)$, and this holds since $x(d+1)= p_{i^\star}(d+1)$.
We get $(\ref{eq:lb}.i^\star)$ analogously.
Now assume $i\ne i^\star$ and let $L_i$ be the left hand side of $(\ref{eq:ub}.i)$. 
Since $x(j) \in \{0, 1\}$ for $j \in [d]$ and $x(d+1) \geq 0$, all the summands of $L_i$ are nonnegative.
Moreover, since $i\ne i^\star$, we have $x(j) \neq \chi_i(j)$ for some $j \in [d]$, and then $L_i \geq t \geq p_i(d+1)$, so the inequality $(\ref{eq:ub}.i)$ holds independently of the value of $x(d+1)$.
Analogously, when $L_i$ is the left hand side of $(\ref{eq:lb}.i)$, we get
$L_i \ge 2t-x(d+1) \ge t \ge t- p_i(d+1)$, as required.
\end{claimproof}

Let $\poly{P}$ be the polytope provided by Claim~\ref{cl:polytope}.
Furthermore, let $q \in \RR^{d+1}$ be the point defined as $q \coloneqq t \cdot \onevector = (t, t, \ldots, t)$.
We consider the instance $\Ic' = (d+1, \poly{P}, q)$ of \PointInCone.
Note that $d = \Oh(\log n)$ and $\encoding(\poly{P}, q) = \Oh(n \log n \cdot \log t)$, which in turn is bounded by $\Oh(n^2\log n)$ due to $t\leq 2^{\Oh(n)}$.
Also, one can easily verify that $\Ic'$ can be computed from $\Ic$ in polynomial time.
Now, we prove that the instance $\Ic'$ is equivalent to $\Ic$.

\begin{claim}
\label{cl:equivalence}
$\Ic$ is a \yesinstance of \SubsetSumRep if and only if $\Ic'$ is a \yesinstance of \PointInCone. 
\end{claim}

\begin{claimproof}
First, assume that $\Ic$ is a \yesinstance of \SubsetSumRep; that is, there are nonnegative integers $\lambda_1, \lambda_2, \ldots, \lambda_n$ such that $\sum_{i=1}^n \lambda_i \cdot a_i = t$.
Our goal is to show that $q \in \IntCone(\poly{P} \cap \ZZ^{d+1}) = \IntCone(P)$. That is, we need to construct a sequence of nonnegative integers $(\lambda_0', \lambda_1', \ldots, \lambda_{2^d-1}')$ such that
\[
\sum_{i=0}^{2^d-1} \lambda_i' \cdot p_i = q.
\]
First, we set $\lambda_i' \coloneqq \lambda_i$, for $i \in [n]$.
Then, we get the required value at the $(d+1)$-st coordinate, i.e.,
\[
\left( \sum_{i=1}^n \lambda_i' \cdot p_i \right)(d+1) = t = q(d+1).
\]
It remains to set the values of $\lambda_i'$ for $i \in[2^d]_0\setminus [n]$.
Note that $p_i(d+1) = 0$ for $i \in [2^d]_0\setminus [n]$, therefore setting those $\lambda_i'$ does not affect the $(d+1)$-st coordinate of the result.

Consider an index $i \in [n]$.
Recall that $\chi_i + \chi_{2^d-i-1} = \onevector$, and since $2^d \geq 2n + 2$, we have $2^d-i-1 \geq 2^d - n-1 \geq n + 1$.
Hence, by setting $\lambda_{2^d-i-1}' \coloneqq \lambda_i' = \lambda_i$, we obtain that
\[
\lambda_i' \cdot p_i + \lambda_{2^d-i-1}' \cdot p_{2^d-i-1} = (\lambda_i, \lambda_i, \ldots, \lambda_i, \lambda_i a_i).
\]
By applying this procedure for every $i \in [n]$, we get a point $q' \in \ZZ^{d+1}$ of the form $(\Lambda, \Lambda, \ldots, \Lambda, t)$,
where
\[
  \Lambda = \sum_{i=1}^n \lambda_i \leq \sum_{i=1}^n \lambda_i \cdot a_i = t.
\]

To obtain the number $t$ on the first $d$ coordinates of the result, it remains to observe that $p_{2^d-1} = (1, 1, \ldots, 1, 0)$, therefore setting $\lambda_{2^d-1}' \coloneqq t - \sum_{i=1}^n \lambda_i$ produces the desired point $q$. (We set $\lambda_i' \coloneqq 0$ for all $i$ not considered in the described procedure.)

For the other direction, suppose that $\Ic'$ is a \yesinstance of \PointInCone, that is, $q \in \IntCone(P)$.
Then, there exist nonnegative integers $\lambda_i$ (for $i \in [2^d]_0)$ such that
\[
\sum_{i=0}^{2^d-1} \lambda_i \cdot p_i = q.
\]
Comparing the $(d+1)$-st coordinate of both sides yields the equality $\sum_{i=1}^n \lambda_i \cdot a_i = t$. This means that $\Ic$ is indeed a \yesinstance of \SubsetSumRep.
\end{claimproof}

Finally, we are ready to prove Theorem~\ref{thm:main}.
Suppose for contradiction that \PointInCone admits an algorithm with running time $\encoding(\poly{P}, q)^{2^{o(d)}}$.
As argued, given an instance $\Ic$ of \SubsetSumRep with $n$ integers and the target integer $t$ bounded by~$2^{\Oh(n)}$, one can in polynomial time compute an equivalent instance $\Ic' = (d, \poly{P}, q)$ of \PointInCone with $d\leq \Oh(\log n)$ and $\encoding(\poly{P},q)\leq \Oh(n^2\log n)$.
Now, running our hypothetical algorithm on $\Ic'$ yields an algorithm for \SubsetSumRep with running~time
\begin{align*}
\label{eq:main-bound}
\encoding(\poly{P},q)^{2^{o(d)}} = (n^2 \log n)^{2^{o(d)}} = (n^2 \log n)^{n^{o(1)}} \leq 2^{n^{o(1)}\cdot 3\log n} \leq 2^{o(n)},
\end{align*}
which contradicts \cref{thm:subsetsumrep}.
This concludes the proof of \cref{thm:main}.

\bibliography{references}

\appendix

\section{Subset Sum with Multiplicities}
\label{app:subsetsumrep}

In this appendix we give a proof of \cref{thm:subsetsumrep}, which we recall here for convenience.

\subsetsumrepthm*

\begin{proof}[Proof of Theorem~\ref{thm:subsetsumrep}] We provide a reduction from \SubsetSum to \SubsetSumRep.
Let $\Ic = (\{a_1, a_2, \ldots, a_n\}, t)$ be the input instance of \SubsetSum.
We may assume that $a_i\leq t$ for all $i\in [n]$.
We construct an equivalent instance $\Ic' = (\{a_1', \ldots, a_n', b_1, \ldots, b_n\}, t')$ of \SubsetSumRep as follows.
The bit encodings of integers $a_i'$, $b_i$ and $t'$ are partitioned into three blocks $B_1, B_2, B_3$, where $B_3$ contains the least significant bits, and $B_1$ the most significant ones.
For an integer $x$ and a block~$B$, we denote by $x |_B$ the integer of bit-length at most $|B|$ consisting of the bits of $x$ at the positions within the block $B$.
The instance $\Ic'$ is defined by the following conditions.
\begin{itemize}
  \item Blocks $B_1$ and $B_3$ are of length $n$, while block $B_2$ is of length $\ceil{\log t}$.
  \item For $i \in [n]$,
  \begin{align*}
    a_i'|_{B_j} = \begin{cases}
      2^{n-i} & \text{for } j = 1,\\
      a_i & \text{for } j = 2, \\
      2^{i-1} & \text{for } j = 3;
    \end{cases}
    &
    \quad\text{ and }\quad
    b_i|_{B_j} = \begin{cases}
      2^{n-i} & \text{for } j = 1,\\
      0 & \text{for } j = 2, \\
      2^{i-1} & \text{for } j = 3.
    \end{cases}
  \end{align*}
  \item The target integer $t'$ is given by
  \[
    t'|_{B_j} = \begin{cases}
      2^n-1 & \text{for } j = 1, \\
      t & \text{for } j = 2, \\
      2^n-1 & \text{for } j = 3. \\
    \end{cases}
  \]
\end{itemize}
Note that the instance $\Ic'$ consists of a set of $n' \coloneqq 2n$ positive integers and a target integer $t' \leq 2^{\Oh(n)}\cdot t$. In particular, if $t\leq 2^{\Oh(n)}$ then also $t'\leq 2^{\Oh(n)}$.
Clearly, $\Ic'$ can be computed from $\Ic$ in polynomial time.
Next, we prove that $\Ic'$ is indeed an instance  equivalent to $\Ic$.

\begin{claim}
$\Ic$ is a \yesinstance of \SubsetSum if and only if $\Ic'$ is a \yesinstance of \SubsetSumRep.
\end{claim}

\begin{claimproof}
($\implies$).
Assume $\Ic$ is a \yesinstance of \SubsetSum.
Let $J \subseteq [n]$ be a~set of indices such that $\sum_{j \in J} a_j = t$.
We construct a~sequence $\lambda_1, \lambda_2, \ldots, \lambda_{2n}$ of $2n$ nonnegative integers as follows.
For $i \in [n]$, we set
\begin{align*}
  \lambda_i = \begin{cases}
  1 & \text{ if } i \in J, \\
  0 & \text{ if } i \not\in J;
  \end{cases}
&
\quad \text{ and }\quad
  \lambda_{n+i} = \begin{cases}
  0 & \text{ if } i \in J, \\
  1 & \text{ if } i \not\in J.
  \end{cases}
\end{align*}
Then it is easy to verify that
\[
  \sum_{i=1}^n \lambda_i \cdot a_i' + \sum_{i=1}^n \lambda_{n+i} \cdot b_i = t',
\]
and thus the sequence $\lambda_1, \lambda_2, \ldots, \lambda_{2n}$ witnesses that $\Ic'$ is a~\yesinstance of \SubsetSumRep.

($\impliedby$).
Assume that $\Ic'$ is a \yesinstance of \SubsetSumRep.
Let $\lambda_1, \lambda_2, \ldots, \lambda_{2n}$ be nonnegative integers such that
\[
  \sum_{i=1}^n \lambda_i \cdot a_i' + \sum_{i=1}^{n} \lambda_{n+i} \cdot b_{i} = t'.
\]
Let $L$ be the left-hand side of the equation above.
Comparing the least significant bit of $L$ and $t'$ yields $\lambda_1 + \lambda_{n+1} \equiv 1 \pmod 2$.
However, if $\lambda_1 + \lambda_{n+1} \geq 2$, then
\[
L|_{B_1} \geq 2 \cdot 2^{n-1} = 2^n > t'|_{B_1},
\]
and consequently $L > t$, which is a contradiction.
Therefore, $\lambda_1 + \lambda_{n+1} = 1$.
Repeating this argument inductively for $i = 2, 3, \ldots, n$ leads us to the conclusion that the equality
\begin{equation}
\label{eq:lambda}
\lambda_i + \lambda_{n+i} = 1
\end{equation}
holds for every $i \in [n]$.
Now, define a set of indices $J \subseteq [n]$ as $J \coloneqq \{i \in [n] \mid \lambda_i = 1\}$.
Then, by comparing $L|_{B_2}$ and $t|_{B_2}$, we must have that
\[
\sum_{j \in J} a_j = t,
\]
since other terms of $L$ do not contribute to $L|_{B_2}$ according to the equation~(\ref{eq:lambda}).
Hence $\Ic$ is a \yesinstance of \SubsetSum, as desired.
\end{claimproof}

We are ready to conclude the proof of~\cref{thm:subsetsumrep}.
Suppose for contradiction there is an algorithm solving \SubsetSumRep in time $2^{o(n')}$ on instances with $n'$ numbers on input and the target integer $t'$ bounded by $2^{\Oh(n')}$.
Then, as explained above, given an instance $\Ic$ of \SubsetSum with $n$ numbers and the target integer $t$ bounded by $2^{\Oh(n)}$, one can in polynomial time compute an equivalent instance $\Ic'$ of \SubsetSumRep with $n'=2n$ numbers and with target $t'\leq 2^{\Oh(n)}\cdot t\leq 2^{\Oh(n)}=2^{\Oh(n')}$.
Running the hypothetical algorithm on $\Ic'$ solves the initial instance $\Ic$ of \SubsetSum in~time
\[
2^{o(n')} = 2^{o(n)},
\]
which contradicts \cref{thm:subsetsum}.
This finishes the proof of \cref{thm:subsetsumrep}.
\end{proof}

\end{document}